\documentclass[english]{article}
\usepackage[T1]{fontenc}
\usepackage[latin9]{inputenc}
\usepackage{algorithm2e}
\usepackage{amsmath}
\usepackage{amsthm}
\usepackage{amssymb}

\makeatletter
\theoremstyle{plain}
\newtheorem{thm}{\protect\theoremname}
\theoremstyle{definition}
\newtheorem{defn}[thm]{\protect\definitionname}
\theoremstyle{plain}
\newtheorem{lem}[thm]{\protect\lemmaname}
\theoremstyle{plain}
\newtheorem{prop}[thm]{\protect\propositionname}
\theoremstyle{definition}
\newtheorem{example}[thm]{\protect\examplename}
\theoremstyle{remark}
\newtheorem{rem}[thm]{\protect\remarkname}

\makeatother

\usepackage{babel}
\providecommand{\definitionname}{Definition}
\providecommand{\examplename}{Example}
\providecommand{\lemmaname}{Lemma}
\providecommand{\propositionname}{Proposition}
\providecommand{\remarkname}{Remark}
\providecommand{\theoremname}{Theorem}

\begin{document}
\title{Completeness of\\
 Unbounded Best-First Game Algorithms}
\author{Quentin Cohen-Solal\\
LAMSADE, Universit\'{e} Paris-Dauphine, PSL, CNRS, France\\
quentin.cohen-solal@dauphine.psl.eu}

\maketitle
\global\long\def\et{\ \wedge\ }%

\global\long\def\terminal{\mathrm{t}}%

\global\long\def\joueur{\mathrm{j}}%

\global\long\def\joueurUn{\mathrm{1}}%

\global\long\def\joueurDeux{\mathrm{2}}%

\global\long\def\fbin{\mathrm{f_{b}}}%

\global\long\def\fterminal{f_{\mathrm{t}}}%

\global\long\def\fadapt{f_{\theta}}%

\global\long\def\actions{\mathcal{A}}%

\global\long\def\etats{\mathcal{S}}%

\global\long\def\Spartiel{\mathcal{S}_{\mathrm{p}}}%

\global\long\def\ubfm{\mathrm{UBFM}}%

\global\long\def\ubfms{\ubfm_{\mathrm{\mathrm{s}}}}%

\global\long\def\umaxn{\ubfm^{n}}%

\global\long\def\maxn{\mathrm{Max}^{n}}%

\global\long\def\umaxns{\ubfms^{n}}%

\global\long\def\descente{\mathrm{descent}}%

\global\long\def\descenten{\descente^{n}}%

\global\long\def\fbinn{\fbin}%

\global\long\def\fterminaln{\fterminal}%

\global\long\def\fadaptn{\fadapt}%

\global\long\def\ou{\,\vee\,}%

\global\long\def\ubfm{\mathrm{UBFM}}%
\global\long\def\argmax{\operatorname*{\mathrm{arg\,max}}}%
\global\long\def\argmin{\operatorname*{\mathrm{arg\,min}}}%
\global\long\def\liste#1#2{\left\{  #1\,|\,#2\right\}  }%
\global\long\def\minimum{\operatorname*{\mathrm{min}}}%

\begin{abstract}
In this article, we prove the completeness of the following game search
algorithms: unbounded best-first minimax with completion and descent
with completion, i.e. we show that, with enough time, they find the
best game strategy. We then generalize these two algorithms in the
context of perfect information multiplayer games. We show that these
generalizations are also complete: they find one of the equilibrium
points.
\end{abstract}

\section{Introduction}

Unbounded Best-First Minimax~\cite{korf1996best,cohen2020minimax,cohen2019apprendre}
is an old game tree search algorithm that is used very little. Unlike
classic Minimax, this one performs a search at an unbounded depth,
making it possible to explore the game tree non-uniformly and thus
to anticipate many turns in advance with regard to the part of the
game tree assumed the most interesting. Most recently, it has been
successfully applied in the context of \emph{reinforcement learning
without knowledge}.~\cite{2020learning,cohen2020minimax}. More precisely,
for many games, at least combined with a search algorithm dedicated
to reinforcement learning, called \emph{Descent}, the unbounded minimax
gives better results than the state of the art of reinforcement learning
without knowledge (i.e. the AlphaZero algorithm~\cite{silver2018general},
based on Monte Carlo Tree Search (MCTS)~\cite{Coulom06,browne2012survey},
another search algorithm, which is stochastic). The Descent algorithm
is a modification of Unbounded Best-first Minimax that builds the
search tree more in depth in order to more effectively propagate endgame
information during learning. Unbounded Minimax and Descent in their
basic version are however not complete in certain contexts like that
of reinforcement learning, that is to say they do not allow to calculate
an equilibrium point of the complete game tree (i.e. a winning strategy
for each player) subject to having sufficient calculation time \cite{2020learning}.
For example, the basic algorithms, not using the fact that certain
states are resolved, will always choose to play during the exploration
a resolved state rather than an unresolved state less well evaluated
(whereas the evaluation of a unresolved state is only an estimate,
so the unresolved state may be better). This then blocks the exploration
and prevents, when this scenario occurs, calculating the minimax value
and therefore determining the best action to play.  In \cite{2020learning},
a modification of Descent and Unbounded Minimax, which is called completion,
has been proposed. 

We show in this article on the one hand that the completion makes
the two algorithms complete (Section~\ref{sec:deux-joueurs}). These
two algorithms being limited to two-player games, we propose in this
article, for each of these two algorithms, two possible generalizations
to the multi-player framework, i.e. to the $\maxn$ framework~\cite{luckhart1986algorithmic,korf1991multi,sturtevant2003last}
the equivalent of Minimax in the multi-player framework. The multiplayer
algorithms that we propose are therefore unbounded versions and variants
of the algorithm $\maxn$ (which searches at a fixed depth). In addition,
we show that these new algorithms are also complete (Section~\ref{sec:multi-joueurs}).

Note that it has been proven that the standard MCTS is complete~\cite{kocsis2006bandit,kocsis2006improved}
and that it has also been generalized to the multi-player framework~\cite{zerbel2019multiagent,petosa2019multiplayer,baier2020guiding}.
Several other search algorithms for multiplayer games have also been
proposed~\cite{sturtevant2000pruning,schadd2011best,nijssen2012overview,fridenfalk2014n,baier2020opponent},
but they are all at fixed depth.

\section{Completeness in Perfect Two Player Games \label{sec:deux-joueurs}}

We start by defining the two player games that we consider in this
article. Then, we recall the Uunbounded Best-First Minimax and Descent
algorithms using the completion technique. Finally, we show that these
two algorithms are indeed complete. 

\subsection{Perfect Two Player Games }

In this section, we are interested in two-player games with perfect
information (games without hidden information or chance where players
take turns playing), that is to say formally: 
\begin{defn}
A perfect two-player game is a tuple $\left(S,\actions,\joueur,\fbin\right)$
where
\begin{itemize}
\item $\left(\etats,\actions\right)$ is a finite directed acyclic graph,
\item $\joueur$ is a function from $\etats$ to $\left\{ 1,2\right\} $,
\item $\fbin$ is a function from $\liste{s\in\etats}{\terminal\left(s\right)}$
to $\left\{ -1,0,1\right\} $, where :
\begin{itemize}
\item $\terminal$ is a predicate such that $\terminal\left(s\right)$ is
true if and only if $\left|\actions\left(s\right)\right|=0$ ;
\item $\actions\left(s\right)$ is the set defined by $\liste{s'\in\etats}{\left(s,s'\right)\in\actions}$,
for all $s\in\etats$.
\end{itemize}
\end{itemize}
\end{defn}

The set $\etats$ is the set of game states. $\actions$ encodes the
actions of the game : $\actions\left(s\right)$ is the set of states
reachable by an action from $s$. The function $\joueur$ indicates,
for each state, the number of the player whose turn it is, i.e. the
player who must play. The predicate $\terminal\left(s\right)$ indicates
if $s$is a terminal state (i.e. an end-of-game state). Let $s\in\etats$
such that $\terminal\left(s\right)$ is true, the value $\fbin\left(s\right)$
is the payout for the first player in the terminal state $s$ (the
gain for the second player is $-\fbin\left(s\right)$ ; we are in
the context of zero-sum games, non-zero-sum games are included in
the multiplayer generalization introduced in the next section). We
have $\fbin\left(s\right)=1$ if the first player wins, $\fbin\left(s\right)=0$
in the event of a draw and $\fbin\left(s\right)=-1$ if the first
player loses. 

The terminal evaluation $\fbin$ is often not very informative about
the quality of a game. A more expressive terminal evaluation function,
with values in $\mathbb{R}$, can be used to favor some games over
others, which can improve the level of play. We denote such functions
by $\fterminal$. For example, in score games, $\fterminal(s)$ maybe
the score of the endgame $s$ (see \cite{2020learning} for other
terminal evaluation functions). In order to guide the search, it is
necessary to be able to evaluate non-terminal states. To do this,
an evaluation function from $\etats$ to $\mathbb{R}$, denoted by
$\fadapt(s)$, is used. A good evaluation function $\fadapt(s)$ provides
for example an approximation of the minimax value of $s$. Such a
function $\fadapt(s)$ can be determined by reinforcement learning,
using $\fterminal(s)$ as ``reinforcement heuristic'' \cite{2020learning}.

\subsection{Algorithms}

\subsubsection{Unbounded Best-First Minimax with Completion }

Unbounded (Best-First) Minimax (noted more succinctly $\ubfm$) is
an algorithm which builds a (partial) tree of the game to then decide
which is the best action to play given the current knowledge about
the game (i.e. given the partial tree of the game). Each state $s$
of the partial game tree is associated with three values. The first
value, the completion value $c\left(s\right)$, indicates the exact
minimax value of $s$ compared to the classic gain of the game (if
it is not known $c\left(s\right)=0$). The (exact) minimax value of
a state $s$ with respect to a certain terminal evaluation function
(such as the classic game gain) is the terminal evaluation of the
end-of-game state obtained by starting from $s$ where the players
play optimally (i.e. the first player maximizes the end-of-game value
and the second player minimizes it ; each by having complete knowledge
of the full game tree). The second value, the heuristic evaluation
$v\left(s\right)$, is an estimate of the true minimax value of $s$
with respect to a certain reinforcement heuristic (a terminal evaluation
function, at least, more expressive than the classic gain of the game).
More precisely, the heuristic evaluation $v\left(s\right)$ is the
minimax value of $s$ in the partial game tree where the terminal
leaves of the tree are evaluated by the reinforcement heuristic and
the other leaves by an adaptive evaluation function (learned, for
example by reinforcement, to estimate the minimax value corresponding
to the reinforcement heuristic). Finally, the third value, the resolution
value $r\left(s\right)$, indicates whether the state $s$ is (weakly)
resolved (case $r\left(s\right)=1$) or not (case $r\left(s\right)=0$).
A terminal state is resolved. A non-terminal leaf is not resolved.
An internal state is resolved if the best child is a winning resolved
state for the current player or if all of its children are resolved
($c\left(s\right)$ is then exact). The algorithm iteratively builds
the partial tree of the game by extending each time the best sequence
of unresolved states (in a state $s$ the first player (resp. second
player) chooses the state $s'$ which maximizes (resp. minimizes)
the lexicographically ordered pair $\left(c\left(s\right),v\left(s\right)\right)$).
In other words, it adds the child states of the principal variation
of the partial game tree deprived of resolved states. An iteration
of Unbounded Minimax is described in Algorithm~\ref{alg:ubfms_search}.
This iteration is performed (from a game state) as long as there is
some search time left ($\tau$ : search time) or until this state
is resolved (and therefore its minimax value is determined, as we
will prove later).

\begin{algorithm}[!bh]
\DontPrintSemicolon\SetAlgoNoEnd

\SetKwFunction{completedbestaction}{completed\_best\_action}\SetKwFunction{completedbestactiondual}{completed\_best\_action\_dual}\SetKwFunction{backupresolution}{backup\_resolution}
\SetKwProg{myproc}{Function}{}{}

\myproc{\completedbestaction{$s$, $A$}}{

\eIf{$\joueur\left(s\right)=\joueurUn$}{return ${\displaystyle \argmax_{s'\in A}\left(c\left(s'\right),v\left(s'\right),n\left(s,s'\right)\right)}$\;}{return
${\displaystyle \argmin_{s'\in A}\left(c\left(s'\right),v\left(s'\right),-n\left(s,s'\right)\right)}$\;}

}

\;

\myproc{\completedbestactiondual{$s$, $A$}}{

\eIf{$\joueur\left(s\right)=\joueurUn$}{return ${\displaystyle \argmax_{s'\in A}\left(c\left(s'\right),v\left(s'\right),-n\left(s,s'\right)\right)}$\;}{return
${\displaystyle \argmin_{s'\in A}\left(c\left(s'\right),v\left(s'\right),n\left(s,s'\right)\right)}$\;}

}

\;

\myproc{\backupresolution{$s$}}{

\eIf{$\left|c\left(s\right)\right|=1$}{return $1$ \;}{return
$\minimum{}_{s'\in\actions\left(s\right)}r\left(s'\right)$\;}

}

\;

\protect\protect

\caption{Definition of the algorithms completed\_best\_action($s$, $A$),
which computes the \emph{a priori }best child state by using completion
from a set of child states $A$, and backup\_resolution($s$), which
updates the resolution of $s$ from its child states. \label{alg:descente-completion-suite}}
\end{algorithm}

\begin{algorithm}[!bh]
\DontPrintSemicolon\SetAlgoNoEnd

\SetKwFunction{completionubfmsiteration}{$\ubfm$\_iteration}\SetKwProg{myproc}{Function}{}{}

\myproc{\completionubfmsiteration{$s$, $\Spartiel$, $T$, $\fadapt$,
$\fterminal$}}{

\eIf{$\terminal\left(s\right)$}{

$\Spartiel\leftarrow\Spartiel\cup\{s\}$\;

$r\left(s\right),c\left(s\right),v\left(s\right)\leftarrow1,\fbin\left(s\right),\fterminal\left(s\right)$

}{

\eIf{$s\notin\Spartiel$}{

$\Spartiel\leftarrow\Spartiel\cup\{s\}$\;

\ForEach{$s'\in\actions\left(s\right)$}{

\eIf{$\terminal\left(s'\right)$}{

$\Spartiel\leftarrow\Spartiel\cup\{s'\}$\;

$r\left(s'\right),c\left(s'\right),v\left(s'\right)\leftarrow1,\fbin\left(s'\right),\fterminal\left(s'\right)$\;

}{

\If{$s'\notin\Spartiel$}{

$r\left(s'\right),c\left(s'\right),v\left(s'\right)\leftarrow0,0,\fadapt\left(s'\right)$\;

}

}

}

}{

$A\leftarrow\liste{s'\in\actions\left(s\right)}{r\left(s'\right)=0}$

\If{$\left|A\right|>0$}{

$s'\leftarrow$ \completedbestactiondual{$s$, $A$}\;

$n(s,s')\leftarrow n(s,s')+1$\;

\completionubfmsiteration{$s'$, $\Spartiel$, $T$, $\fadapt$,
$\fterminal$}}\;

$s'\leftarrow$ \completedbestaction{$s$, $\actions\left(s\right)$}\;

$c(s),v(s)\leftarrow c\left(s'\right),v\left(s'\right)$\;

$r\left(s\right)\leftarrow$ \backupresolution{$s$}\;

}}

}

\;

\protect\protect

\caption{\emph{$\protect\ubfm$} tree search algorithm with completion (see
Section~\ref{sec:deux-joueurs} for the definitions of symbols ;
see Algorithm~\ref{alg:descente-completion-suite} for the definitions
of completed\_best\_action($s$) and backup\_resolution($s$)). Note:
$T=(v,c,r)$.\label{alg:ubfms_search}}
\end{algorithm}

\subsubsection{Descent with Completion}

Descent is a variant of Unbounded Minimax. Its difference is that
it will extend the best sequence of unresolved game states until it
reaches a terminal state or a resolved state. So unlike Unbounded
Minimax, it can add the children of several states at each iteration.
It explores the game tree in a different order, much more in depth
first (this exploration priority is interesting from a learning point
of view~\cite{2020learning}). An iteration of Descent is described
in Algorithm~\ref{alg:descente-completion}.

\begin{algorithm}[!bh]
\DontPrintSemicolon\SetAlgoNoEnd

\SetKwFunction{completiondescenteiteration}{$\descente$\_iteration}
\SetKwProg{myproc}{Function}{}{}

\myproc{\completiondescenteiteration{$s$, $\Spartiel$, $T$,
$\fadapt$, $\fterminal$}}{

\eIf{$\terminal\left(s\right)$}{

$\Spartiel\leftarrow\Spartiel\cup\{s\}$\;

$r\left(s\right),c\left(s\right),v\left(s\right)\leftarrow1,\fbin\left(s\right),\fterminal\left(s\right)$

}{

\If{$s\notin\Spartiel$}{

$\Spartiel\leftarrow\Spartiel\cup\{s\}$\;

\ForEach{$s'\in\actions\left(s\right)$}{

\eIf{$\terminal\left(s'\right)$}{

$\Spartiel\leftarrow\Spartiel\cup\{s'\}$\;

$r\left(s'\right),c\left(s'\right),v\left(s'\right)\leftarrow1,\fbin\left(s'\right),\fterminal\left(s'\right)$\;

}{

\If{$s'\notin\Spartiel$}{

$r\left(s'\right),c\left(s'\right),v\left(s'\right)\leftarrow0,0,\fadapt\left(s'\right)$\;

}

}

}

$s'\leftarrow$ \completedbestaction{$s$, $\actions\left(s\right)$}\;

$c\left(s\right),v\left(s\right)\leftarrow c\left(s'\right),v\left(s'\right)$\;

$r\left(s\right)\leftarrow$ \backupresolution{$s$}\;

}

\If{$r\left(s\right)=0$}{

$A\leftarrow\liste{s'\in\actions\left(s\right)}{r\left(s'\right)=0}$

$s'\leftarrow$ \completedbestactiondual{$s$, $A$}\;

$n(s,s')\leftarrow n(s,s')+1$\;

\completiondescenteiteration{$s'$ , $\Spartiel$, $T$, $\fadapt$,
$\fterminal$}\;

$s'\leftarrow$ \completedbestaction{$s$, $\actions\left(s\right)$}\;

$c\left(s\right),v\left(s\right)\leftarrow c\left(s'\right),v\left(s'\right)$\;

$r\left(s\right)\leftarrow$ \backupresolution{$s$}\;

}

}

}

\;

\protect\protect

\caption{\emph{Descent} tree search algorithm with completion (see Section~\ref{sec:deux-joueurs}
for the definitions of symbols and Algorithm~\ref{alg:descente-completion-suite}
for the definitions of completed\_best\_action($s$) and backup\_resolution($s$)).
Note: $T=(v,c,r)$.\label{alg:descente-completion}}
\end{algorithm}

\subsubsection{Proof of Completeness }

We now show that the two algorithms are complete. We start by formalizing
precisely what we call being complete in the context of these two
algorithms and then we establish the completeness result. 
\begin{defn}
The minimax value of a state $s\in\etats$ (compared to the terminal
evaluation $\fbin$ in the complete game tree $\etats$) is the value
$M\left(s\right)$ recursively defined by 
\[
M\left(s\right)=\begin{cases}
\max_{s'\in\actions\left(s\right)}M\left(s'\right) & \text{si }\lnot\terminal\left(s\right)\et\joueur\left(s\right)=\joueurUn\\
\min_{s'\in\actions\left(s\right)}M\left(s'\right) & \text{si }\lnot\terminal\left(s\right)\et\joueur\left(s\right)=\joueurDeux\\
\fbin\left(s\right) & \text{si }\terminal\left(s\right)
\end{cases}
\]
\end{defn}

\begin{lem}
\label{lem:Si-avant}If $r\left(s\right)=1$ before an iteration of
$\descente$ (resp. $\ubfm$) then after the iteration , $r\left(s\right)$
and $c\left(s\right)$ have not changed.
\end{lem}

\begin{lem}
\label{lem:completion-resolution}Let $\left(\Spartiel,\actions\right)$
be a game tree built by the algorithm $\ubfm$ or by the algorithm
$\descente$ from a certain state. Let $s\in\Spartiel$. We have
the following properties: 
\begin{itemize}
\item If $r\left(s\right)=1$, then either $\left|c\left(s\right)\right|=1$
or for all $s'\in\actions\left(s\right)$, $r\left(s'\right)=1$ ;
\item If $\left|c\left(s\right)=1\right|$, then $r\left(s\right)=1$ ;
\item If for all $s'\in\actions\left(s\right)$, $r\left(s'\right)=1$,
then $r\left(s\right)=1$.
\end{itemize}
\end{lem}

\begin{proof}
By definition of the algorithm (in particular by the definition and
by the use of the method backup\_resolution($s$) and because assoon
as $r\left(s\right)=1$, $r\left(s\right)$ and $c\left(s\right)$
do not change anymore (Lemma~\ref{lem:Si-avant})).
\end{proof}
\begin{prop}
\label{prop:exact}Let $\left(\etats,\actions\right)$ be a game tree
built by $\ubfm$ or $\descente$ from a certain state and let $s\in\etats$. 

If $r\left(s\right)=1$, then $c\left(s\right)=M\left(s\right)$.
\end{prop}

\begin{proof}
Let $\left(\Spartiel,\actions\right)$ be a game tree built by $\ubfm$
(resp. $\descente$) from a certain state (i.e. the algorithm has
been applied $k$ times on that state). We show by induction this
property. Let $s\in\Spartiel$ such that $r\left(s\right)=1$. We
first show that this property is true for terminal states. 

Suppose in addition that $\terminal\left(s\right)$ is true. Thus,
$c\left(s\right)=\fbin\left(s\right)$. Consequently, $c\left(s\right)=\fbin\left(s\right)=M\left(s\right)$.

We now show this property for non-terminal states: therefore suppose
instead that $\terminal\left(s\right)$ is false.
\begin{itemize}
\item Suppose on the one hand that $\joueur\left(s\right)=\joueurUn$.
\end{itemize}
If $r\left(s\right)=1$, then either $\left|c\left(s\right)\right|=1$
or for all $s'\in\actions\left(s\right)$, $r\left(s'\right)=1$,
by Lemma~\ref{lem:completion-resolution}. If $c\left(s\right)=-1$,
then, at the iteration that calculated $c\left(s\right)=-1$, we had:
for all $s'\in\actions\left(s\right)$, $c\left(s'\right)=-1$ (as
at this moment $c\left(s\right)=\max_{s'\in\actions\left(s\right)}c\left(s'\right)$).
Therefore, since this iteration, for all $s'\in\actions\left(s\right)$,
$r\left(s'\right)=1$ (by Lemma~\ref{lem:completion-resolution}
and Lemma~\ref{lem:Si-avant}). Thus, we have two cases: either $c\left(s\right)=1$
or for all $s'\in\actions\left(s\right)$, $r\left(s'\right)=1$.

If for all $s'\in\actions\left(s\right)$, $r\left(s'\right)=1$,
then, by induction, we have for all $s'\in\actions\left(s\right)$,
$c\left(s'\right)=M\left(s'\right)$. But $c\left(s\right)=\max_{s'\in\actions\left(s\right)}c\left(s'\right)$.
Therefore, $c\left(s\right)=\max_{s'\in\actions\left(s\right)}M\left(s'\right)$,
hence $c\left(s\right)=M\left(s\right)$.

If $c\left(s\right)=1$, then there exists $\tilde{s}\in\actions\left(s\right)$
such that $c\left(\tilde{s}\right)=c\left(s\right)$ at the iteration
calculating $c\left(s\right)=1$ (as $c\left(s\right)=\max_{s'\in\actions\left(s\right)}c\left(s'\right)$
at this iteration). Since we had $c\left(\tilde{s}\right)=1$ at this
iteration, we also had $r\left(\tilde{s}\right)=1$, by Lemma~\ref{lem:completion-resolution}.
Therefore, we always have $c\left(\tilde{s}\right)=c\left(s\right)$
and $r\left(\tilde{s}\right)=1$. Moreover, by induction, $c\left(\tilde{s}\right)=M\left(\tilde{s}\right)$,
hence $c\left(s\right)=M\left(\tilde{s}\right)$. However, since $M\left(\tilde{s}\right)=1$,
we have $M\left(\tilde{s}\right)\geq M\left(s'\right)$ for all $s'\in\actions\left(s\right)$.
Thus, $c\left(s\right)=\max_{s'\in\actions\left(s\right)}M\left(s'\right)$,
hence $c\left(s\right)=M\left(s\right)$.
\begin{itemize}
\item Now suppose on the other hand, instead that $\joueur\left(s\right)=\joueurDeux$. 
\end{itemize}
If $r\left(s\right)=1$ then either $\left|c\left(s\right)\right|=1$
or for all $s'\in\actions\left(s\right)$, $r\left(s'\right)=1$,
by Lemma~\ref{lem:completion-resolution}. If $c\left(s\right)=1$,
then at the iteration that calculated $c\left(s\right)=1$, we had:
for all $s'\in\actions\left(s\right)$, $c\left(s'\right)=1$ (as
at this moment $c\left(s\right)=\min_{s'\in\actions\left(s\right)}c\left(s'\right)$).
Therefore, since this iteration, for all $s'\in\actions\left(s\right)$,
$r\left(s'\right)=1$ (by Lemma~\ref{lem:completion-resolution}
and Lemma~\ref{lem:Si-avant}). Thus, we have two cases: either $c\left(s\right)=-1$
or for all $s'\in\actions\left(s\right)$, $r\left(s'\right)=1$.

If for all $s'\in\actions\left(s\right)$, $r\left(s'\right)=1$,
then, by induction, we have for all $s'\in\actions\left(s\right)$,
$c\left(s'\right)=M\left(s'\right)$. But $c\left(s\right)=\min_{s'\in\actions\left(s\right)}c\left(s'\right)$.
Therefore $c\left(s\right)=\min_{s'\in\actions\left(s\right)}M\left(s'\right)$,
hence $c\left(s\right)=M\left(s\right)$.

If $c\left(s\right)=-1$, then there exists $\tilde{s}\in\actions\left(s\right)$
such that $c\left(\tilde{s}\right)=c\left(s\right)$ at the iteration
calculating $c\left(s\right)=-1$ (as $c\left(s\right)=\min_{s'\in\actions\left(s\right)}c\left(s'\right)$
at this iteration). Since we had $c\left(\tilde{s}\right)=-1$ at
this iteration, we also had $r\left(\tilde{s}\right)=1$, by Lemma~\ref{lem:completion-resolution}.
Therefore, we always have $c\left(\tilde{s}\right)=c\left(s\right)$
and $r\left(\tilde{s}\right)=1$. Moreover, by induction, $c\left(\tilde{s}\right)=M\left(\tilde{s}\right)$,
hence $c\left(s\right)=M\left(\tilde{s}\right)$. However, since $M\left(\tilde{s}\right)=-1$,
we have $M\left(\tilde{s}\right)\leq M\left(s'\right)$ for all $s'\in\actions\left(s\right)$.
Thus, $c\left(s\right)=\min_{s'\in\actions\left(s\right)}M\left(s'\right)$,
hence $c\left(s\right)=M\left(s\right)$.
\end{proof}
\begin{prop}
\label{prop:arret}Let $\etats$ be the set of states of a two-player
game. It exists $N\in\mathbb{N}$ such that after applying $N$ times
the algorithm $\descente$ (resp. $\ubfm$) on a state $s\in\etats$,
we have $r\left(s\right)=1$.
\end{prop}

\begin{proof}
We show that with at most $N=2\left|\etats\right|$ iterations of
$\descente$ (resp. $\ubfm$) applied on a same state $s\in\etats$,
we have $r\left(s\right)=1$. Note first that if $s$ is terminal
or satisfies $r\left(s\right)=1$. then after having applied the algorithm,
we have $r\left(s\right)=1$. Now suppose that $s$ is not terminal
and satisfies $r\left(s\right)=0$. To show this property we show
that each iteration adds in $\Spartiel$ at least one state of $\etats$
not being in $\Spartiel$ or marks as ``solved'' an additional state,
that is, a state $s'\in\etats$ satisfying $r\left(s'\right)=0$,
satisfies $r\left(s'\right)=1$ after the iteration. This is sufficient
to show the property, because either after one of the iterations,
we have $r\left(s\right)=1$, or the iterative application of the
algorithm ends up adding in $\Spartiel$ all descendants of $s$ and/or
by marking as ``solved'' all states of $\Spartiel$. Indeed, if
all the descendants of $s$ are added then necessarily $r\left(s\right)=1$
(since by induction all descendants verify $r\left(s\right)=1$ ;
by definition and use of backup\_resolution($s$)). Since $S$ is
finite, with at most $2\left|\etats\right|$ iterations, $r\left(s\right)=1$. 

We therefore show, under the assumptions $r\left(s\right)=0$ and
$\lnot\terminal\left(s\right)$, that each iteration adds at least
one new state of $S$ in $\Spartiel$ or change the value $r\left(s'\right)$
from $0$ to $1$ for a some state $s'\in\etats$. Let $\tilde{s}$
be the current state analyzed by the algorithm (at the beginning $\tilde{s}=s$).
If $\tilde{s}$ is not in $\Spartiel$, then $\tilde{s}$ is added
in $\Spartiel$. Otherwise for $\ubfm$ and then for $\descente$,
the algorithm recursively chooses the best child of the current state
satisfying $r\left(\tilde{s}'\right)=0$, which we denote $\tilde{s}'$.
For $\ubfm$, this recursion is performed until $\tilde{s}$ is not
in $\Spartiel$ (and adds it ) or that $\tilde{s}$ is terminal or
that there is no child state $\tilde{s}'$ satisfying $r\left(\tilde{s}'\right)=0$.
Given that $\tilde{s}$ necessarily satisfies $r\left(\tilde{s}\right)=0$
at the start of each recursion, $\tilde{s}$ is not terminal. Therefore,
this recursion is performed until $\tilde{s}$ is not in $\Spartiel$
or that there is no child state $\tilde{s}'$ satisfying $r\left(\tilde{s}'\right)=0$.
In the latter case, all the children $\tilde{s}'$ of the state $\tilde{s}$
satisfy $r\left(\tilde{s}'\right)=1$, and thus, at the end of the
iteration we have $r\left(\tilde{s}\right)=1$ while at the beginning
we have $r\left(\tilde{s}\right)=0$. Therefore, with $\ubfm$, each
iteration effectively adds a new state in $\Spartiel$ or mark as
solved a new state. With $\descente$ this recursion is performed
until the state $\tilde{s}$ is terminal or satisfies $r\left(\tilde{s}\right)=1$
after the block of the test ``$\tilde{s}\in\Spartiel$''. Note that
with $\descente,$ if $r\left(\tilde{s}\right)=0$ after the block
of the test ``$\tilde{s}\in\Spartiel$'', then there still exists
a child state $\tilde{s}'$ satisfying $r\left(\tilde{s}'\right)=0$
(otherwise the block changes the value of $\tilde{s}$ to $r\left(\tilde{s}\right)=1$).
Since $r\left(\tilde{s}\right)=0$ at the start of each descent recursion
step (and that $s$ is not terminal), the recursion is performed until
the state $\tilde{s}$ satisfies $r\left(\tilde{s}\right)=1$ after
the test $\tilde{s}\in\Spartiel$. Thus, when the iteration ends,
before the test, we have $r\left(\tilde{s}\right)=0$ and after the
test, we have $r\left(\tilde{s}\right)=1$. Therefore, necessarily
before the test, $\tilde{s}$ was not in $\Spartiel$ and thus $\tilde{s}$
has been added. Consequently, for the two algorithms, an iteration
adds at least one new state of $\etats$ in $\Spartiel$ or marks
as resolved a new state (under the assumption $s$ is neither terminal
nor solved).
\end{proof}
\begin{thm}
The algorithm $\descente$ and the algorithm $\ubfm$ are complete,
i.e. applying $\descente$ (resp. $\ubfm$) on a state $s\in\etats$,
with a search time $\tau$ large enough, gives $r\left(s\right)=1$
and $c\left(s\right)=M\left(s\right)$.
\end{thm}

\begin{proof}
By Proposition~\ref{prop:arret}, then by Proposition~\ref{prop:exact}.
\end{proof}

\section{Perfect Multiplayer Games\label{sec:multi-joueurs}}

We now generalize to the multiplayer framework the two algorithms,
unbounded minimax and descent with completion, to obtain respectively
Unbounded $\maxn$ (that we note more succinctly $\umaxn$) and $\descenten$
(both with completion). We propose two possible generalizations for
each of the two algorithms and we show that the $4$ algorithms are
all complete. Note that there is a difference between the two generalizations
only if the game have draws.

In the multi-player framework (or two players with non-zero sum),
the gain of the first player is no longer sufficient to characterize
the end of the game. Each player must therefore have their own end-of-game
gain. A terminal state is therefore evaluated by a $n$-uplet of values
where the $i$-th component is the gain of $i$-th player and $n$
is the total number of players. The goal of each player is then obviously
to maximize its final gain. Unbounded $\maxn$ and $\descenten$ consist,
in the context of each of the two variants, in iteratively extending
the best sequence of unresolved states. For this, each player maximizes
its own gain, that is to say it plays the state $s'$ from the state
$s$ which lexicographically maximizes the pair $\left(c\left(s'\right)_{\joueur\left(s\right)},v\left(s'\right)_{\joueur\left(s\right)}\right)$
until reaching a leaf of the tree for Unbounded $\maxn$and and until
reaching a terminal state or a resolved state for $\descenten$. As
in the two-player framework, during each iteration, states that are
not in the partial game tree are added.

However, generalizing to the multiplayer framework is not so easy,
as a fundamental property is lost. Unlike the two-player framework,
if the best child completion value is non-zero, then the current state
is not necessarily resolved. Formally, the property $\argmax_{s'\in\actions\left(s\right)}c\left(s'\right)_{\joueur\left(s\right)}\neq\left(0,\ldots,0\right)\Rightarrow r\left(s\right)=1$
is lost (or $r\left(s\right)=1$ no longer means to be solved). It
is therefore necessary to modify the way in which either the completion
value or the resolution value is calculated. Each of the two possibilities
leads to a different algorithm, which builds a different game tree.
The first possibility is to propagate the best completion value to
the current state, only if the current state is resolved. With this
approach, we get a weaker property $c\left(s\right)\neq\left(0,\ldots,0\right)\implies r\left(s\right)=1$,
but we also lose the weaker property $c\left(s\right)=\argmax_{s'\in\actions\left(s\right)}c\left(s'\right)_{\joueur\left(s\right)}$.
The other possibility (introduced later) is to always propagate the
best completion value, which requires that a state is marked as solved
only if its best child is marked as solved. With this variant, we
do not have the property $c\left(s\right)\neq\left(0,\ldots,0\right)\implies r\left(s\right)=1$,
but we recover the property $c\left(s\right)=\argmax_{s'\in\actions\left(s\right)}c\left(s'\right)_{\joueur\left(s\right)}$.
It is not clear whether one of the two algorithms is better than the
other. With the second approach, one can use exact information about
what is likely to happen (instead of being limited to an estimate
of what will happen). This information is the value $c\left(s\right)$,
which is, with this variant, the $\maxn$ value in the partial game
tree based on the classic terminal evaluation, i.e. the completion
value of the last state of the best states sequence (which is terminal
if $c\left(s\right)\neq\left(0,\ldots,0\right)$ ; in this case $c\left(s\right)$
is the $n$-uplet of end-of-game gains). This additional information,
although being the exact value of a terminal state, is not necessarily
the exact value of that state. This happens if one of the players
prefers to have the guarantee of a draw to a possibility of losing
the game.
\begin{example}
The first player has the choice between two states $s'$ and $s''$
in a certain state $s$. The state $s'$, which is terminal, is evaluated
by $c\left(s'\right)=\left(0,0,0,-1\right)$ and $v\left(s'\right)=\left(0,0,0,-1\right)$.
The state $s''$ is evaluated by $c\left(s''\right)=\left(0,0,0,0\right)$
and $v\left(s''\right)=\left(-1,1,1,1\right)$. With the current state
of knowledge about the game, a reasonable choice for the first player
is to choose the state $s'$ (choose the state $s''$ is also a reasonable
choice but it is a risky choice which seems less interesting although
it could actually be more interesting). With the first variant, since
the state $s$ is not resolved, $s$ is evaluated by $c\left(s\right)=\left(0,\ldots,0\right)$.
With the second variant, $c\left(s\right)=\left(0,1,0,-1\right)$,
although $s$ is still not resolved.
\end{example}

With the first approach, this information, the completion value $c\left(s\right)$,
which may be misleading or advantageous, is not used to distinguish
unresolved states. Note that with the second approach, it is necessary
to adapt the calculation of the best child to favor the resolved states
over the unresolved states with the same completion value.

In the multi-player framework, there is another property that is lost,
it is the uniqueness of the completion value of a state. This property
can however be recovered under a certain assumption. In particular,
it is necessary to consider the equilibrium point with respect to
the pairs $\left(c\left(s\right),v\left(s\right)\right)$, i.e. set
the value $\maxn$ of a state $s$ as the pair $\left(\fbin\left(s'\right),\fterminal\left(s'\right)\right)$
of the last state $s'$ of the best states sequence starting from
the state $s$ in the complete game tree. In addition, the terminal
evaluation $\fterminal$ must verify a particular property, which
guarantees to be able to decide between two terminal states of different
gains. Otherwise, without these two conditions, two states can have
the same gain for a same player but different gains for the other
players. Thus, without these two conditions, there would be multiple
optimal strategies with an uncertain outcome (i.e. a \textquotedblleft king-making\textquotedblright{}
effect).

\subsection{Definition\label{subsec:Definition-multi}}

We are now interested in multi-player perfect information games, that
is to say formally:
\begin{defn}
A perfect multiplayer game with $n$ players ($n>1$), is a tuple
$\left(S,\actions,\joueur,\fbin\right)$ where
\begin{itemize}
\item $\left(\etats,\actions\right)$ is a finite acyclic directed graph,
\item $\joueur$ is a function from $\etats$ to $\left\{ 1,\ldots,n\right\} $,
\item $\fbin$ is a function from $\liste{s\in\etats}{\terminal\left(s\right)}$
to $\left\{ -1,0,1\right\} ^{n}$, with :
\begin{itemize}
\item $\terminal$ a predicate such as $\terminal\left(s\right)$ is true
if and only if $\left|\actions\left(s\right)\right|=0$ and
\item $\actions\left(s\right)$ the set defined by $\liste{s'\in\etats}{\left(s,s'\right)\in\actions}$,
for all $s\in\etats$.
\end{itemize}
\end{itemize}
\end{defn}

The set $\etats$ is the set of game states. $\actions$ encodes the
actions of the game : $\actions\left(s\right)$ is the set of states
reachable by an action from $s$. The function $\joueur$ indicates,
for each state, the number of the player whose turn it is, that is
to say the player who must play. The predicate $\terminal\left(s\right)$
indicates if $s$ is a terminal state (i.e. an end-of-game state).
Let $s\in\etats$ such that $\terminal\left(s\right)$ is true. The
value $\fbin\left(s\right)_{j}$ is the gain for the $j$-th player
in the terminal state $s$. We have $\fbin\left(s\right)_{j}=1$ if
the $j$-th player is winning, $\fbin\left(s\right)_{j}=0$ in the
event of a draw for the $j$-th player and $\fbin\left(s\right)_{j}=-1$
if the $j$-th player is losing.

The terminal evaluation $\fbin$ is often not very informative about
the quality of a game. A more expressive terminal evaluation function,
valuable in $\mathbb{R}^{n}$, can be used to favor some games to
others, which can improve the level of play. We denote such functions
by $\fterminal$. For example, in score games, $\fterminal(s)$ can
be the endgame scores of each player in the terminal state $s$. In
order to guide the search of the best action, it is necessary to be
able to evaluate non-terminal states. To do this, an evaluation function
from $\etats$ to $\mathbb{R}^{n}$, denoted by $\fadapt(s)$, is
used. A good evaluation function $\fadapt(s)$ provides for example
an approximation of the value $\maxn$ of $s$. Such a function $\fadapt(s)$
can be determined by reinforcement learning, by using $\fterminal(s)$
as a ``reinforcement heuristic'' and using the descent framework
\cite{2020learning} with the algorithm $\descenten$ instead of $\descente$.

\subsection{First Multi-player Generalization }

We start by introducing the first generalization, which consists in
modifying $c\left(s\right)$ only if $s$ is resolved and which introduces
an additional evaluation $c'\left(s\right)$. The evaluation $c'\left(s\right)$
is calculated from the values $c\left(s'\right)$ of the children
$s'$ of $s$ in the same way as $c\left(s\right)$ in the classic
case. If $s$ is resolved, then $c\left(s\right)=c'\left(s\right)$
and $c\left(s\right)$ is the value $\maxn$ in the partial game tree
whose leaves are labeled by the classic game gain $\fbinn$. Otherwise,
$c\left(s\right)=\left(0,\ldots,0\right)$. This ensures that if $c\left(s\right)$
(resp. $c'\left(s\right)$) is not zero, then this value corresponds
to a resolved child. To calculate an equilibrium point (an optimal
strategy) in the multi-player case, it is also necessary that a state
be considered resolved if it has a winning child state for the current
player and the value $v\left(s'\right)$ of this child state is maximum
(or that all the children are resolved). An iteration of Unbounded
$\maxn$ is described in Algorithm~\ref{alg:umaxn_iteration}. Once
the partial game tree has been built, as in the two-player framework
\cite{2020learning}, there are two strategies for deciding which
action to play: the one leading to the best value state (we choose
the child $s'$ of $s$ which lexicographically maximizes $\left(c\left(s'\right)_{\joueur\left(s\right)},v\left(s'\right)_{\joueur\left(s\right)}\right)$)
(see Algorithm~\ref{alg:umaxn}) or the safest action (we choose
the action leading to a winning state of higher value if it exists,
otherwise we choose the child $s'$ of $s$ which maximizes the number
of times it has been selected from $s$ and since the start of the
game (see Algorithm~\ref{alg:umaxns}) . An iteration of $\descenten$
is described in Algorithm~\ref{alg:descente-n_iteration} and the
complete code of $\descenten$ is described in Algorithm~\ref{alg:descente_n}.

\begin{algorithm}[!bh]
\DontPrintSemicolon\SetAlgoNoEnd

\SetKwFunction{completedbestactionn}{best\_action\_n}\SetKwFunction{completedbestactionndual}{best\_action\_n\_dual}\SetKwProg{myproc}{Function}{}{}

\myproc{\completedbestactionn{$s$, $T$}}{

return ${\displaystyle \argmax_{s'\in\actions\left(s\right)}\left(c\left(s'\right)_{\joueur\left(s\right)},v\left(s'\right)_{\joueur\left(s\right)},n\left(s,s'\right)\right)}$\;

}

\;

\myproc{\completedbestactionndual{$s$, $T$}}{

return ${\displaystyle \argmax_{s'\in\actions\left(s\right)}\left(c\left(s'\right)_{\joueur\left(s\right)},v\left(s'\right)_{\joueur\left(s\right)},-n\left(s,s'\right)\right)}$\;

}

\;

\caption{Best action computation for $n$ players (see Section~\ref{subsec:Definition-multi}
for the definitions of symbols).\label{alg:best_action_n}}
\end{algorithm}

\begin{algorithm}[!bh]
\DontPrintSemicolon\SetAlgoNoEnd

\SetKwFunction{backupresolutionn}{backup\_resolution\_n} \SetKwProg{myproc}{Function}{}{}

\myproc{\backupresolutionn{$s$}}{

\eIf{$c'\left(s\right)_{\joueur\left(s\right)}=1\et v\left(s\right)_{\joueur\left(s\right)}=\max_{s'\in\liste{s'\in S}{\terminal\left(s'\right)}}\fterminaln\left(s'\right)_{\joueur\left(s\right)}$}{return
$1$\;}{return $\minimum{}_{s'\in\actions\left(s\right)}r\left(s'\right)$\;}

}

\;

\protect\protect

\caption{Definition of backup\_resolution\_n($s$), which updates the resolution
value of the state $s$ from its child states. \label{alg:backup_n}}
\end{algorithm}

\begin{algorithm}[!bh]
\DontPrintSemicolon\SetAlgoNoEnd

\SetKwFunction{completionubfmsiteration}{$\umaxn$\_iteration}

\SetKwFunction{backupresolutionn}{backup\_resolution\_n}\SetKwProg{myproc}{Function}{}{}

\myproc{\completionubfmsiteration{$s$, $\Spartiel$, $T$, $\fadaptn$,
$\fterminaln$}}{

\eIf{$\terminal\left(s\right)$}{

$\Spartiel\leftarrow\Spartiel\cup\{s\}$\;

$r\left(s\right),c\left(s\right),v\left(s\right)\leftarrow1,\fbinn\left(s\right),\fterminaln\left(s\right)$

}{

\eIf{$s\notin\Spartiel$}{

$\Spartiel\leftarrow\Spartiel\cup\{s\}$\;

\ForEach{$s'\in\actions\left(s\right)$}{

\eIf{$\terminal\left(s'\right)$}{

$\Spartiel\leftarrow\Spartiel\cup\{s'\}$\;

$r\left(s'\right),c\left(s'\right),v\left(s'\right)\leftarrow1,\fbinn\left(s'\right),\fterminaln\left(s'\right)$\;

}{

\If{$s'\notin\Spartiel$}{

$r\left(s'\right),c\left(s'\right),v\left(s'\right)\leftarrow0,\left(0,\ldots,0\right),\fadaptn\left(s'\right)$\;

}

}

}

}{

$A\leftarrow\liste{s'\in\actions\left(s\right)}{r\left(s'\right)=0}$

\If{$\left|A\right|>0$}{

$s'\leftarrow$ \completedbestactionndual{$s$, $A$}\;

$n(s,s')\leftarrow n(s,s')+1$\;

 \completionubfmsiteration{$s'$, $\Spartiel$, $T$, $\fadaptn$,
$\fterminaln$}}\;

$s'\leftarrow$ \completedbestactionn{$s$, $\actions\left(s\right)$}\;

$c'(s),v(s)\leftarrow c\left(s'\right),v\left(s'\right)$\;

$r\left(s\right)\leftarrow$ \backupresolutionn{$s$}\;

\If{$r\left(s\right)$}{

$c(s)\leftarrow c'\left(s\right)$\;

}}

}

}

\;

\protect\protect

\caption{Iteration algorithm of $\protect\umaxn$ with completion (see Section~\ref{subsec:Definition-multi}
for the definitions of symbols, Algorithm~\ref{alg:best_action_n}
for the definitions of completed\_best\_action\_n($s$) and Algorithm~\ref{alg:backup_n}
for the definitions of backup\_resolution\_n($s$)). Note: $T=(v,c,r)$,
each $c\left(s\right)$ is initialized to $\left(0,\ldots,0\right)$
and each number of selection of $s'$ from $s$ $n(s,s')$ is initialized
to $0$.\label{alg:umaxn_iteration}}
\end{algorithm}
\begin{algorithm}[!bh]
\DontPrintSemicolon\SetAlgoNoEnd

\SetKwFunction{decisioncompletionumaxn}{\emph{$\umaxn$}}\SetKwFunction{time}{time}\SetKwProg{myproc}{Function}{}{}

\;

\myproc{\decisioncompletionumaxn{$s$, $\Spartiel$, $T$, $\fadaptn$,
$\fterminaln$, $\tau$}}{

$t=$ \time{}\;

\lWhile{\time{}$-\,t<\tau\wedge r\left(s\right)=0$}{\completionubfmsiteration{$s$,
$\Spartiel$, $T$, $\fadaptn$, $\fterminaln$}}

return \completedbestactionn{$s$, $T$}\;

}\;

\caption{The algorithm\emph{ $\protect\umaxn$} (see Section~\ref{subsec:Definition-multi}
for the definitions of symbols ; see Algorithm~\ref{alg:umaxn_iteration}
for the code of $\protect\umaxn$\_iteration($s$, $\protect\Spartiel$,
$T$, $\protect\fadaptn$, $\protect\fterminaln$) ; time() returns
the current time in seconds ; $\tau$ : search time per action).\label{alg:umaxn}}
\end{algorithm}
\begin{algorithm}[!bh]
\DontPrintSemicolon\SetAlgoNoEnd

\SetKwFunction{decisioncompletionumaxns}{\emph{$\umaxns$}}\SetKwFunction{time}{time}\SetKwFunction{completedsafestactionn}{safest\_action\_n}\SetKwProg{myproc}{Function}{}{}

\myproc{\completedsafestactionn{$s$, $T$}}{

return ${\displaystyle \argmax_{s'\in\actions\left(s\right)}\left(c\left(s'\right)_{\joueur\left(s\right)},n\left(s,s'\right),v\left(s'\right)_{\joueur\left(s\right)}\right)}$\;

}

\;

\myproc{\decisioncompletionumaxns{$s$, $\Spartiel$, $T$, $\fadaptn$,
$\fterminaln$, $\tau$}}{

$t=$ \time{}\;

\lWhile{\time{}$-\,t<\tau\wedge r\left(s\right)=0$}{\completionubfmsiteration{$s$,
$\Spartiel$, $T$, $\fadaptn$, $\fterminaln$}}

return \completedsafestactionn{$s$, $T$}\;

}\;

\caption{The algorithm\emph{ $\protect\umaxns$} (see Section~\ref{subsec:Definition-multi}
for the definitions of symbols ; see Algorithm~\ref{alg:umaxn_iteration}
for the code of $\protect\umaxn$\_iteration($s$, $\protect\Spartiel$,
$T$, $\protect\fadaptn$, $\protect\fterminaln$) ; time() returns
the current time in seconds ; $\tau$ : search time per action).\label{alg:umaxns}}
\end{algorithm}

\begin{algorithm}[!bh]
\DontPrintSemicolon\SetAlgoNoEnd

\SetKwFunction{completiondescenteiteration}{$\descenten$\_iteration}

\SetKwFunction{backupresolutionn}{backup\_resolution\_n} \SetKwProg{myproc}{Function}{}{}

\myproc{\completiondescenteiteration{$s$, $\Spartiel$, $T$,
$\fadaptn$, $\fterminaln$}}{

\eIf{$\terminal\left(s\right)$}{

$\Spartiel\leftarrow\Spartiel\cup\{s\}$\;

$r\left(s\right),c\left(s\right),v\left(s\right)\leftarrow1,\fbinn\left(s\right),\fterminaln\left(s\right)$

}{

\If{$s\notin\Spartiel$}{

$\Spartiel\leftarrow\Spartiel\cup\{s\}$\;

\ForEach{$s'\in\actions\left(s\right)$}{

\eIf{$\terminal\left(s'\right)$}{

$\Spartiel\leftarrow\Spartiel\cup\{s'\}$\;

$r\left(s'\right),c\left(s'\right),v\left(s'\right)\leftarrow1,\fbinn\left(s'\right),\fterminaln\left(s'\right)$\;

}{

\If{$s'\notin\Spartiel$}{

$r\left(s'\right),c\left(s'\right),v\left(s'\right)\leftarrow0,\left(0,\ldots,0\right),\fadaptn\left(s'\right)$\;

}

}

}

$s'\leftarrow$ \completedbestactionn{$s$, $\actions\left(s\right)$}\;

$c'(s),v(s)\leftarrow c\left(s'\right),v\left(s'\right)$\;

$r\left(s\right)\leftarrow$ \backupresolutionn{$s$}\;

\If{$r\left(s\right)$}{

$c(s)\leftarrow c'\left(s\right)$\;

}

}

\If{$r\left(s\right)=0$}{

$A\leftarrow\liste{s'\in\actions\left(s\right)}{r\left(s'\right)=0}$

$s'\leftarrow$ \completedbestactionndual{$s$, $A$}\;

$n(s,s')\leftarrow n(s,s')+1$\;

\completiondescenteiteration{$s'$, $\Spartiel$, $T$, $\fadaptn$,
$\fterminaln$}\;

$s'\leftarrow$ \completedbestactionn{$s$, $\actions\left(s\right)$}\;

$c'(s),v(s)\leftarrow c\left(s'\right),v\left(s'\right)$\;

$r\left(s\right)\leftarrow$ \backupresolutionn{$s$}\;

\If{$r\left(s\right)$}{

$c(s)\leftarrow c'\left(s\right)$\;

}

}}

}

\;

\protect\protect

\caption{Iteration algorithm of\emph{ $\protect\descenten$} with completion
(see Section~\ref{subsec:Definition-multi} for the definitions of
symbols, Algorithm~\ref{alg:best_action_n} for the definitions of
completed\_best\_action\_n($s$) and Algorithm~\ref{alg:backup_n}
for the definitions of backup\_resolution\_n($s$)). Note: $T=(v,c,r)$
and each $c\left(s\right)$ is initialized to $\left(0,\ldots,0\right)$.\label{alg:descente-n_iteration}}
\end{algorithm}

\begin{algorithm}[!bh]
\DontPrintSemicolon\SetAlgoNoEnd

\SetKwFunction{decisioncompletionudescenten}{\emph{$\descenten$}}\SetKwFunction{completiondescenten}{iteration\_\emph{$\descenten$}}\SetKwFunction{time}{time}\SetKwFunction{bestexplorationactionn}{best\_exploration\_action\_n}\SetKwFunction{completedbestexplorationactionn}{completed\_best\_exploration\_action\_n}\SetKwProg{myproc}{Function}{}{}

\myproc{\completedbestexplorationactionn{$s$, $T$}}{

\eIf{$r\left(s\right)$}{

return ${\displaystyle \argmax_{s'\in\actions\left(s\right)}\left(c\left(s'\right)_{\joueur\left(s\right)},v\left(s'\right)_{\joueur\left(s\right)}\right)}$\;

}{

return \bestexplorationactionn{$s$, $T$}\;

}

}

\;

\myproc{\decisioncompletionudescenten{$s$, $\Spartiel$, $T$,
$\fadaptn$, $\fterminaln$, $\tau$}}{

$t=$ \time{}\;

\lWhile{\time{}$-\,t<\tau\wedge r\left(s\right)=0$}{\completiondescenten{$s$,
$\Spartiel$, $T$, $\fadaptn$, $\fterminaln$}}

return \completedbestexplorationactionn{$s$, $T$}\;

}\;

\caption{The\emph{ $\protect\descenten$} algorithm (see Section~\ref{subsec:Definition-multi}
for the definitions of symbols ; see Algorithm~\ref{alg:descente-n_iteration}
for the code of \emph{$\protect\descenten$}\_iteration($s$, $S$,
$T$, $\protect\fadaptn$, $\protect\fterminaln$) ; best\_exploration\_action\_n$\left(s,T\right)$
is an exploration method, for example an action selection distribution
(see~\cite{2020learning})).\label{alg:descente_n}}
\end{algorithm}

\begin{rem}
Some variants of $\umaxn$ and $\descenten$ are possible, which perhaps
have a practical interest, at the cost of completeness. First, we
can consider that a state is resolved as soon as $c'\left(s\right)_{\joueur\left(s\right)}=1$
(instead of imposing in addition $v\left(s\right)_{\joueur\left(s\right)}=\max_{s'\in\liste{s'\in S}{\terminal\left(s'\right)}}\fterminaln\left(s\right)_{\joueur\left(s\right)}$).
Another variant is to stop an iteration of descent when the best action
leads to a draw, i.e. choosing $s'\leftarrow$ \completedbestactionn{$s$,
$\actions\left(s\right)$} instead of $s'\leftarrow$ \completedbestactionn{$s$,
$A$}. This could be interesting in the context of games with a lot
of draws.

Regarding the choice of the action to play after having carried out
a possibly partial search, an alternative to the best action can be
\[
\argmax_{s'\in\actions\left(s\right)}\left(c\left(s'\right)_{\joueur\left(s\right)},c'\left(s'\right)_{\joueur\left(s\right)},v\left(s'\right)_{\joueur\left(s\right)},n\left(s,s'\right)\right)
\]
 and an alternative to the safest action may be 
\[
\argmax_{s'\in\actions\left(s\right)}\left(c\left(s'\right)_{\joueur\left(s\right)},c'\left(s'\right)_{\joueur\left(s\right)},n\left(s,s'\right),v\left(s'\right)_{\joueur\left(s\right)}\right).
\]

Finally, note that the criterion $v\left(s\right)_{\joueur\left(s\right)}=\max_{s'\in\liste{s'\in S}{\terminal\left(s'\right)}}\fterminaln\left(s\right)_{\joueur\left(s\right)}$,
can be replaced, while keeping the completeness, by a criterion of
local maximum compared to the terminal states which are descendants
of $s$, if the local maximum is known, i.e. $v\left(s\right)_{\joueur\left(s\right)}=\max_{s'\in D_{\mathrm{t}}\left(s\right)}\fterminaln\left(s'\right)_{\joueur\left(s\right)}$
where 
\[
D_{\mathrm{t}}\left(s\right)=\liste{s'\in S}{\terminal\left(s'\right)\et\exists s_{1},\ldots,s_{k}\in\etats\,s'\in\actions\left(s_{k}\right)\wedge s_{k}\in\actions\left(s_{k-1}\right)\wedge\cdots\wedge s_{2}\in\actions\left(s_{1}\right)\wedge s_{1}\in\actions\left(s\right)}
\]
.
\end{rem}

\subsubsection{Proof of Completeness}

We now show that the two algorithms are complete. We start by formalizing
precisely what we call being complete in the context of these two
algorithms and then we establish the completeness result.
\begin{defn}
Let $\left(\etats,\actions,\terminal,\joueur,\fbin\right)$ be a perfect
multi-player game. Let $\fterminal$ be a terminal evaluation function
for this game.

The evaluation $\fterminaln$ is said to be tie-breaking for the game
if for all $s,s'\in\etats$ such that $\terminal\left(s\right)$ and
$\terminal\left(s'\right)$ , we have either for all $j\in\left\{ 1,\ldots,n\right\} $,
$\fbinn\left(s\right)_{j}=\fbinn\left(s'\right)_{j}\implies\fterminaln\left(s\right)_{j}\neq\fterminaln\left(s'\right)_{j}$
or $\fbinn\left(s\right)=\fbinn\left(s'\right)\et\fterminaln\left(s\right)=\fterminaln\left(s'\right)$.
\end{defn}

\begin{defn}
The value $\maxn$ of a state $s\in\etats$ with respect to the terminal
evaluation $\fterminaln$ (in the complete game tree $\etats$) is
the value $M\left(s\right)$ recursively defined by

\[
M\left(s\right)=\begin{cases}
M\left(\argmax_{s'\in\actions\left(s\right)}\left(M\left(s'\right)_{0,\joueur\left(s\right)},M\left(s'\right)_{1,\joueur\left(s\right)}\right)\right) & \text{si }\lnot\terminal\left(s\right)\\
\left(\fbin\left(s\right),\fterminaln\left(s\right)\right) & \text{si }\terminal\left(s\right)
\end{cases}.
\]
\end{defn}

\begin{lem}
Let $\left(\Spartiel,\actions\right)$ be a game tree built by the
algorithm $\umaxn$ or by the algorithm $\descenten$ from a certain
state and a tie-breaking terminal evaluation function for the game.
Let $s\in\Spartiel$ .\label{lem:Si-avant-n}

If $r\left(s\right)=1$ before an iteration of $\descente$ (resp.
$\ubfm$) then after the iteration, $r\left(s\right)$, $c\left(s\right)$
and $v\left(s\right)$ have not changed.
\end{lem}

\begin{prop}
\label{lem:completion-resolution-n-1-1}Let $\left(\Spartiel,\actions\right)$
be a game tree built by the algorithm $\umaxn$ or by the algorithm
$\descenten$ from a certain state and a tie-breaking terminal evaluation
function for the game. Let $s\in\Spartiel$. We have the following
property:
\begin{itemize}
\item if $r\left(s\right)=1$ then either $c'\left(s\right)_{\joueur\left(s\right)}=1$
and $v\left(s\right)_{\joueur\left(s\right)}=\max_{s'\in\liste{s'\in S}{\terminal\left(s'\right)}}\fterminaln\left(s'\right)_{\joueur\left(s\right)}$
or for all $s'\in\actions\left(s\right)$, $r\left(s'\right)=1$ ;
\item if $c'\left(s\right)_{\joueur\left(s\right)}=1$ and $v\left(s\right)_{\joueur\left(s\right)}=\max_{s'\in\liste{s'\in S}{\terminal\left(s'\right)}}\fterminaln\left(s'\right)_{\joueur\left(s\right)}$
then $r\left(s\right)=1$ ;
\item if for all $s'\in\actions\left(s\right)$, $r\left(s'\right)=1$ then
$r\left(s\right)=1$.
\end{itemize}
\end{prop}

\begin{proof}
By definition of the algorithm (in particular by the definition and
by the use of the method backup\_resolution\_n($s$) and because as
soon as we have $r\left(s\right)=1$, $r\left(s\right)$, $v\left(s\right)$,
$c'\left(s\right)$ and $c\left(s\right)$ do not change anymore (Lemma~\ref{lem:Si-avant-n})).
\end{proof}
\begin{lem}
\label{lem:resolution-si-component-max}Let $\left(\Spartiel,\actions\right)$
be a game tree built by the algorithm $\umaxn$ or by the algorithm
$\descenten$ from a certain state. Let $s\in\Spartiel$ . 

If there exists $j$ such that $c\left(s\right)_{j}=1$ then $r\left(s\right)=1$.
\end{lem}

\begin{proof}
By definition of the algorithm, we have $r\left(s\right)=0\implies c\left(s\right)=\left(0,\ldots,0\right)$.
\end{proof}
\begin{prop}
\label{prop:exact_n}Let $\fterminaln$ be a tie-breaking terminal
evaluation function. Let $\left(\Spartiel,\actions\right)$ be a game
tree built by $\umaxn$ or $\descenten$ from a certain state using
$\fterminal$. Let $s\in\Spartiel$. 

If $r\left(s\right)=1$ then there exists a unique value $\maxn$
of $s$with respect to $\fterminal$, denoted by $M\left(s\right)$,
and we have $\left(c\left(s\right),v\left(s\right)\right)=M\left(s\right)$.
\end{prop}

\begin{proof}
Let $\left(\Spartiel,\actions\right)$ be a game tree built by $\umaxn$
(resp. $\descenten$) from a certain state (i.e. the algorithm has
been applied $k$ times on that state). We show this property by induction.
Let $s\in\Spartiel$ such that $r\left(s\right)=1$. We first show
that this property holds for terminal states.

Suppose in addition that $\terminal\left(s\right)$ is true. Thus,
$\left(c\left(s\right),v\left(s\right)\right)=\left(\fbin\left(s\right),\fterminaln\left(s\right)\right)$
and therefore $\left(c\left(s\right),v\left(s\right)\right)=\left(\fbin\left(s\right),\fterminaln\left(s\right)\right)=M\left(s\right)$.

We now show this property for non-terminal states: we suppose instead
that $\terminal\left(s\right)$ is false.

Since $r\left(s\right)=1$, we have either for all $s'\in\actions\left(s\right)$,
$r\left(s'\right)=1$ or $c'\left(s\right)_{\joueur\left(s\right)}=1$
and $v\left(s\right)_{\joueur\left(s\right)}=\max_{s'\in\liste{s'\in S}{\terminal\left(s'\right)}}\fterminaln\left(s'\right)_{\joueur\left(s\right)}$,
by Lemma~\ref{lem:completion-resolution-n-1-1}. We also have $c\left(s\right)=c'\left(s\right)$.

If for all $s'\in\actions\left(s\right)$, $r\left(s'\right)=1$,
then by induction, we have for all $s'\in\actions\left(s\right)$,
$\left(c\left(s'\right),v\left(s'\right)\right)=M\left(s'\right)$.
But $c\left(s\right)=c\left(\argmax_{s'\in\actions\left(s\right)}\left(c\left(s'\right)_{\joueur\left(s\right)},v\left(s'\right)_{\joueur\left(s\right)}\right)\right)$
and $v\left(s\right)=v\left(\argmax_{s'\in\actions\left(s\right)}\left(c\left(s'\right)_{\joueur\left(s\right)},v\left(s'\right)_{\joueur\left(s\right)}\right)\right)$,
since there is a unique pair $\left(c\left(s'\right),v\left(s'\right)\right)$
maximizing $\left(c\left(s'\right)_{\joueur\left(s\right)},v\left(s'\right)_{\joueur\left(s\right)}\right)$
(as $\fterminaln$ is tie-breaking and by Lemma~\ref{lem:Si-avant-n}).
Therefore $\left(c\left(s\right),v\left(s\right)\right)=M\left(\argmax_{s'\in\actions\left(s\right)}\left(M\left(s'\right)_{0,\joueur\left(s\right)},M\left(s'\right)_{1,\joueur\left(s\right)}\right)\right)$,
hence $\left(c\left(s\right),v\left(s\right)\right)=M\left(s\right)$.

If $c'\left(s\right)_{\joueur\left(s\right)}=1$ and $v\left(s\right)_{\joueur\left(s\right)}=\max_{s'\in\liste{s'\in S}{\terminal\left(s'\right)}}\fterminaln\left(s'\right)_{\joueur\left(s\right)}$,
there exists $\tilde{s}\in\actions\left(s\right)$ such that we had
$c'\left(s\right)=c\left(\tilde{s}\right)$ and $v\left(s\right)=v\left(\tilde{s}\right)$
at the iteration that marked $s$ as resolved. Thus, at this iteration
$c\left(\tilde{s}\right)_{\joueur\left(s\right)}=1$ and therefore
$r\left(\tilde{s}\right)=1$, by Lemma~\ref{lem:resolution-si-component-max}.
Thus, we still have $c'\left(s\right)=c\left(\tilde{s}\right)$, $v\left(s\right)=v\left(\tilde{s}\right)$
and $r\left(\tilde{s}\right)=1$ (Lemma~\ref{lem:Si-avant-n}). By
induction, $\left(c\left(\tilde{s}\right),v\left(\tilde{s}\right)\right)=M\left(\tilde{s}\right)$
and therefore $\left(c\left(s\right),v\left(s\right)\right)=\left(c'\left(s\right),v\left(s\right)\right)=M\left(\tilde{s}\right)$.
But, since $M\left(\tilde{s}\right)$ is maximum for the player $\joueur\left(s\right)$,
for all $s'\in\actions\left(s\right)$ either $\left(M\left(\tilde{s}\right)_{0,\joueur\left(s\right)},M\left(\tilde{s}\right)_{1,\joueur\left(s\right)}\right)>\left(M\left(s'\right)_{0,\joueur\left(s\right)},M\left(s'\right)_{1,\joueur\left(s\right)}\right)$
or $M\left(\tilde{s}\right)=M\left(s'\right)$ (as $\fterminaln$
is tie-breaking). Thus, 
\[
\left(c\left(s\right),v\left(s\right)\right)=M\left(\argmax_{s'\in\actions\left(s\right)}\left(M\left(s'\right)_{0,\joueur\left(s\right)},M\left(s'\right)_{1,\joueur\left(s\right)}\right)\right).
\]
Hence $\left(c\left(s\right),v\left(s\right)\right)=M\left(s\right)$.
\end{proof}

\begin{prop}
\label{prop:arret_n}Let $S$ be the set of states of a perfect multi-player
game. There exists $N\in\mathbb{N}$ such that after applying $N$
times the algorithm $\descenten$ (resp. $\umaxn$) on any state $s\in S$,
we have $r\left(s\right)=1$.
\end{prop}

\begin{proof}
We show that with at most $N=2\left|\etats\right|$ iterations of
$\descenten$ (resp. $\umaxn$) applied to a certain state $s\in\etats$,
we have $r\left(s\right)=1$. Note first that if $s$ is terminal
or satisfies $r\left(s\right)=1$. then after applying the algorithm,
we have $r\left(s\right)=1$. Now suppose that $s$ is not terminal
and satisfies $r\left(s\right)=0$. To show the proposition, we show
that each iteration adds in $\Spartiel$ at least one state of $\etats$
which is not in $\Spartiel$ or marks as solved an additional state,
i.e. a state $s'\in\etats$ satisfying $r\left(s'\right)=0$, satisfies
$r\left(s'\right)=1$ after the iteration. This is sufficient to show
the property, because either after one of the iterations, we have
$r\left(s\right)=1$ or the iterative application of the algorithm
ends up adding in $\Spartiel$ all descendants of $s$ and/or by marking
all states of $\Spartiel$ as resolved. Indeed, if all the descendants
of $s$ are added then necessarily $r\left(s\right)=1$ (since by
induction all descendants satisfy $r\left(s\right)=1$ ; by definition
and use of backup\_resolution\_n($s$)). Since $S$ is finite, with
at most $2\left|\etats\right|$ iterations, $r\left(s\right)=1$. 

We therefore show, under the assumption $r\left(s\right)=0$ et $\lnot\terminal\left(s\right)$,
that each iteration adds at least one new state of $S$ in $\Spartiel$
or change the value $r\left(s'\right)$ from $0$ to $1$ for a certain
state $s'\in\etats$. Let $\tilde{s}$ be the current state analyzed
by the algorithm (at the beginning $\tilde{s}=s$). If $\tilde{s}$
is not in $\Spartiel$, then $\tilde{s}$ is added in $\Spartiel$.
Otherwise for $\ubfm$ and then for $\descente$, the algorithm recursively
chooses the best child of the current state satisfying $r\left(\tilde{s}'\right)=0$,
which we denote $\tilde{s}'$. For $\umaxn$, this recursion is performed
until $\tilde{s}$ is not in $\Spartiel$ (and adds it) or that $\tilde{s}$
is terminal or that there is no child $\tilde{s}'$ satisfying $r\left(\tilde{s}'\right)=0$.
Given that $\tilde{s}$ necessarily satisfies $r\left(\tilde{s}\right)=0$
at the beginning of each recursion, $\tilde{s}$ is not terminal.
Therefore, this recursion is performed until $\tilde{s}$ is not in
$\Spartiel$ or that there is no child $\tilde{s}'$ satisfying $r\left(\tilde{s}'\right)=0$.
In the latter case, all the children $\tilde{s}'$ of the state $\tilde{s}$
satisfies $r\left(\tilde{s}'\right)=1$, and therefore at the end
of the iteration, we have $r\left(\tilde{s}\right)=1$ while at the
beginning we have $r\left(\tilde{s}\right)=0$. Thus, with $\umaxn$,
each iteration adds a new state in $\Spartiel$ or marks as solved
a new state. With $\descenten$, this recursion is performed until
the state $\tilde{s}$ is terminal or satisfies $r\left(\tilde{s}\right)=1$
after the block of the test ``$\tilde{s}\in\Spartiel$''. Note that
with $\descente,$ if $r\left(\tilde{s}\right)=0$ after the block
of the test ``$\tilde{s}\in\Spartiel$'', then there is always a
child $\tilde{s}'$ satisfying $r\left(\tilde{s}'\right)=0$ (otherwise
the block would have changed the value of $\tilde{s}$ to $r\left(\tilde{s}\right)=1$).
Since $r\left(\tilde{s}\right)=0$ at the start of each descent recursion
step (and that $s$ is not terminal), this recursion is performed
until the state $\tilde{s}$ satisfies $r\left(\tilde{s}\right)=1$
after the block of the test $\tilde{s}\in\Spartiel$. Thus, when this
iteration ends, before the test, we have $r\left(\tilde{s}\right)=0$
and after the test, we have $r\left(\tilde{s}\right)=1$. Therefore,
necessarily before the test, $\tilde{s}$ is not in $\Spartiel$ and
therefore $\tilde{s}$ is added. Thus, for the two algorithms, an
iteration adds at least one new state of $\etats$ in $\Spartiel$
marks as solved a new state (under the assumption that $s$ is neither
terminal nor solved).
\end{proof}
\begin{thm}
The algorithm $\descenten$ and the algorithm $\umaxn$ are ``complete'',
i.e. applying $\descenten$ (resp. $\umaxn$) on any state $s\in\etats$
by using a tie-breaking terminal evaluation $\fterminal$, with a
search time $\tau$ large enough , gives $r\left(s\right)=1$ and
$\left(c\left(s\right),v\left(s\right)\right)=M\left(s\right)$, the
unique value $\maxn$ of $s$ with respect to $\fterminal$.
\end{thm}

\begin{proof}
By Proposition~\ref{prop:arret_n}, then by Proposition~\ref{prop:exact_n}.
\end{proof}

\subsection{Second Multi-player Generalization}

We now introduce the second generalization, which allows to keep the
property $c\left(s\right)=\argmax_{s'\in\actions\left(s\right)}c\left(s'\right)_{\joueur\left(s\right)}$
and to use the additional information $c\left(s\right)$ about unresolved
states to build the partial game tree and decide on the best action
to play.

\subsubsection{Algorithms}

With the second generalization, the two algorithms $\descenten$ and
Unbounded $\maxn$ are analogous to the algorithms of the first generalization
but with several differences described below. At any time $c\left(s\right)=c'\left(s\right)$
(so there is no need to $c'$). In addition, as with this variant
an unresolved state can have a non-zero completion value, it is necessary
to be able to separate a winning resolved state from an unresolved
\textquotedblleft winning\textquotedblright{} state. Thus, the calculation
of the best action consists in choosing the child state $s'$ of $s$
maximizing $\left(r\left(s'\right)\cdot c\left(s'\right),c\left(s'\right),v\left(s'\right)\right)$
(Algorithm~\ref{alg:best_action-var}). In addition, choosing the
safest action then amounts to maximizing $\left(r\left(s'\right)\cdot c\left(s'\right),c\left(s'\right),n\left(s,s'\right),v\left(s'\right)\right)$
(Algorithm~\ref{alg:safe_n_var}). Finally, a state is resolved if
all its children are resolved or if $c\left(s\right)_{\joueur\left(s\right)}=1$,
if $v\left(s\right)_{\joueur\left(s\right)}$ is maximum, and if there
is a solved child $s'$ such that $\left(c\left(s\right),v\left(s\right)\right)=\left(c\left(s'\right),v\left(s'\right)\right)$
(Algorithm~\ref{alg:backup_n-var}). The code of an iteration of
Unbounded $\maxn$ in the context of this variant is given in Algorithm~\ref{alg:umaxn-var}.
The code of an iteration of $\descenten$ in the context of this variant
is given in Algorithm~\ref{alg:descente-n-var}. 

\begin{algorithm}[!bh]
\DontPrintSemicolon\SetAlgoNoEnd

\SetKwFunction{completedbestactionn}{best\_action\_n}\SetKwFunction{completedbestactionndual}{best\_action\_n\_dual}\SetKwProg{myproc}{Function}{}{}

\myproc{\completedbestactionn{$s$, $T$}}{

return ${\displaystyle \argmax_{s'\in\actions\left(s\right)}\left(r\left(s'\right)\cdot c\left(s'\right)_{\joueur\left(s\right)},c\left(s'\right)_{\joueur\left(s\right)},v\left(s'\right)_{\joueur\left(s\right)},n\left(s,s'\right)\right)}$\;

}

\;

\myproc{\completedbestactionndual{$s$, $T$}}{

return ${\displaystyle \argmax_{s'\in\actions\left(s\right)}\left(r\left(s'\right)\cdot c\left(s'\right)_{\joueur\left(s\right)},c\left(s'\right)_{\joueur\left(s\right)},v\left(s'\right)_{\joueur\left(s\right)},-n\left(s,s'\right)\right)}$\;

}

\;

\caption{Best action function for $n$ players (see Section~\ref{subsec:Definition-multi}
for the definitions of symbols).\label{alg:best_action-var}}
\end{algorithm}

\begin{algorithm}[!bh]
\DontPrintSemicolon\SetAlgoNoEnd

\SetKwFunction{backupresolutionn}{backup\_resolution\_n} \SetKwProg{myproc}{Function}{}{}

\myproc{\backupresolutionn{$s$}}{

\eIf{$c\left(s\right)_{\joueur\left(s\right)}=1\et v\left(s\right)_{\joueur\left(s\right)}=\max_{s'\in\liste{s'\in S}{\terminal\left(s'\right)}}\fterminaln\left(s'\right)_{\joueur\left(s\right)}\et\exists\tilde{s}\in\actions\left(s\right)\ c\left(s\right)=c\left(\tilde{s}\right)\et v\left(s\right)=v\left(\tilde{s}\right)\et r\left(\tilde{s}\right)=1$}{return
$1$\;}{return $\minimum{}_{s'\in\actions\left(s\right)}r\left(s'\right)$\;}

}

\;

\protect\protect

\caption{Definition of backup\_resolution\_n($s$), which updates the resolution
value of the state $s$ from its child states. \label{alg:backup_n-var}}
\end{algorithm}

\begin{algorithm}[!bh]
\DontPrintSemicolon\SetAlgoNoEnd

\SetKwFunction{decisioncompletionumaxns}{\emph{$\umaxns$}}\SetKwFunction{time}{time}\SetKwFunction{completedsafestactionn}{safest\_action\_n}\SetKwProg{myproc}{Function}{}{}

\myproc{\completedsafestactionn{$s$, $T$}}{

return ${\displaystyle \argmax_{s'\in\actions\left(s\right)}\left(r\left(s'\right)\cdot c\left(s'\right)_{\joueur\left(s\right)},c\left(s'\right)_{\joueur\left(s\right)},n\left(s,s'\right),v\left(s'\right)_{\joueur\left(s\right)}\right)}$\;

}

\;

\caption{Safest action computation (see Section~\ref{subsec:Definition-multi}
for the definitions of symbols).\label{alg:safe_n_var}}
\end{algorithm}

\begin{algorithm}[!bh]
\DontPrintSemicolon\SetAlgoNoEnd

\SetKwFunction{completionubfmsiteration}{$\umaxn$\_iteration}

\SetKwFunction{backupresolutionn}{backup\_resolution\_n}\SetKwProg{myproc}{Function}{}{}

\myproc{\completionubfmsiteration{$s$, $\Spartiel$, $T$, $\fadaptn$,
$\fterminaln$}}{

\eIf{$\terminal\left(s\right)$}{

$\Spartiel\leftarrow\Spartiel\cup\{s\}$\;

$r\left(s\right),c\left(s\right),v\left(s\right)\leftarrow1,\fbinn\left(s\right),\fterminaln\left(s\right)$

}{

\eIf{$s\notin\Spartiel$}{

$\Spartiel\leftarrow\Spartiel\cup\{s\}$\;

\ForEach{$s'\in\actions\left(s\right)$}{

\eIf{$\terminal\left(s'\right)$}{

$\Spartiel\leftarrow\Spartiel\cup\{s'\}$\;

$r\left(s'\right),c\left(s'\right),v\left(s'\right)\leftarrow1,\fbinn\left(s'\right),\fterminaln\left(s'\right)$\;

}{

\If{$s'\notin\Spartiel$}{

$r\left(s'\right),c\left(s'\right),v\left(s'\right)\leftarrow0,\left(0,\ldots,0\right),\fadaptn\left(s'\right)$\;

}

}

}

}{

$A\leftarrow\liste{s'\in\actions\left(s\right)}{r\left(s'\right)=0}$

\If{$\left|A\right|>0$}{

$s'\leftarrow$ \completedbestactionndual{$s$, $A$}\;

$n(s,s')\leftarrow n(s,s')+1$\;

 \completionubfmsiteration{$s'$, $\Spartiel$, $T$, $\fadaptn$,
$\fterminaln$}}\;

$s'\leftarrow$ \completedbestactionn{$s$, $\actions\left(s\right)$}\;

$c(s),v(s)\leftarrow c\left(s'\right),v\left(s'\right)$\;

$r\left(s\right)\leftarrow$ \backupresolutionn{$s$}\;

}}

}

\;

\protect\protect

\caption{Iteration algorithm of\emph{ $\protect\umaxn$} with completion (see
Section~\ref{subsec:Definition-multi} for the definitions of symbols,
Algorithm~\ref{alg:best_action-var} for the definitions of completed\_best\_action\_n($s$)
and Algorithm~\ref{alg:backup_n-var} for the definitions of backup\_resolution\_n($s$)).
Note: $T=(v,c,r)$.\label{alg:umaxn-var}}
\end{algorithm}

\begin{algorithm}[!bh]
\DontPrintSemicolon\SetAlgoNoEnd

\SetKwFunction{completiondescenteiteration}{$\descenten$\_iteration}

\SetKwFunction{backupresolutionn}{backup\_resolution\_n} \SetKwProg{myproc}{Function}{}{}

\myproc{\completiondescenteiteration{$s$, $\Spartiel$, $T$,
$\fadaptn$, $\fterminaln$}}{

\eIf{$\terminal\left(s\right)$}{

$\Spartiel\leftarrow\Spartiel\cup\{s\}$\;

$r\left(s\right),c\left(s\right),v\left(s\right)\leftarrow1,\fbinn\left(s\right),\fterminaln\left(s\right)$

}{

\If{$s\notin\Spartiel$}{

$\Spartiel\leftarrow\Spartiel\cup\{s\}$\;

\ForEach{$s'\in\actions\left(s\right)$}{

\eIf{$\terminal\left(s'\right)$}{

$\Spartiel\leftarrow\Spartiel\cup\{s'\}$\;

$r\left(s'\right),c\left(s'\right),v\left(s'\right)\leftarrow1,\fbinn\left(s'\right),\fterminaln\left(s'\right)$\;

}{

\If{$s'\notin\Spartiel$}{

$r\left(s'\right),c\left(s'\right),v\left(s'\right)\leftarrow0,\left(0,\ldots,0\right),\fadaptn\left(s'\right)$\;

}

}

}

$s'\leftarrow$ \completedbestactionn{$s$, $\actions\left(s\right)$}\;

$c(s),v(s)\leftarrow c\left(s'\right),v\left(s'\right)$\;

$r\left(s\right)\leftarrow$ \backupresolutionn{$s$}\;

}

\If{$r\left(s\right)=0$}{

$A\leftarrow\liste{s'\in\actions\left(s\right)}{r\left(s'\right)=0}$

$s'\leftarrow$ \completedbestactionndual{$s$, $A$}\;

$n(s,s')\leftarrow n(s,s')+1$\;

\completiondescenteiteration{$s'$, $\Spartiel$, $T$, $\fadaptn$,
$\fterminaln$}\;

$s'\leftarrow$ \completedbestactionn{$s$, $\actions\left(s\right)$}\;

$c(s),v(s)\leftarrow c\left(s'\right),v\left(s'\right)$\;

$r\left(s\right)\leftarrow$ \backupresolutionn{$s$}\;

}

}}

\;

\protect\protect

\caption{Iteration algorithm of\emph{ $\protect\descenten$} with completion
(Section~\ref{subsec:Definition-multi} for the definitions of symbols,
Algorithm~\ref{alg:best_action-var} for the definitions of completed\_best\_action\_n($s$),
and Algorithm~\ref{alg:backup_n-var} for the definitions of backup\_resolution\_n($s$)).
Note: $T=(v,c,r)$.\label{alg:descente-n-var}}
\end{algorithm}

\subsubsection{Proof of Completeness}

We now show that the two algorithms of the second variant are complete.
\begin{lem}
\label{lem:completion-resolution-n-1-1-var}Let $\left(\Spartiel,\actions\right)$
be a game tree built by the algorithm $\umaxn$ or by the algorithm
$\descenten$ from a certain state. Let $s\in\Spartiel$. We have
the following property:
\begin{itemize}
\item if $r\left(s\right)=1$ then either $c\left(s\right)_{\joueur\left(s\right)}=1$
and $v\left(s\right)_{\joueur\left(s\right)}=\max_{s'\in\liste{s'\in S}{\terminal\left(s'\right)}}\fterminaln\left(s'\right)_{\joueur\left(s\right)}$
and there exists $\tilde{s}\in\actions\left(s\right)$ such that $c\left(s\right)=c\left(\tilde{s}\right)$,
$v\left(s\right)=v\left(\tilde{s}\right)$ and $r\left(\tilde{s}\right)=1$
or for all $s'\in\actions\left(s\right)$, $r\left(s'\right)=1$ ;
\item if $c\left(s\right)_{\joueur\left(s\right)}=1$ and $v\left(s\right)_{\joueur\left(s\right)}=\max_{s'\in\liste{s'\in S}{\terminal\left(s'\right)}}\fterminaln\left(s'\right)_{\joueur\left(s\right)}$
and there exists $\tilde{s}\in\actions\left(s\right)$ such that $c\left(s\right)=c\left(\tilde{s}\right)$,
$v\left(s\right)=v\left(\tilde{s}\right)$ and $r\left(\tilde{s}\right)=1$
then $r\left(s\right)=1$ ;
\item if for all $s'\in\actions\left(s\right)$, $r\left(s'\right)=1$ then
$r\left(s\right)=1$.
\end{itemize}
\end{lem}

\begin{proof}
By definition of the algorithm (in particular by the definition and
by the use of the method backup\_resolution\_n($s$) and because as
soon as we have $r\left(s\right)=1$, $r\left(s\right)$, $v\left(s\right)$
and $c\left(s\right)$ do not change anymore)).
\end{proof}
\begin{prop}
\label{prop:exact_n-var}Let $\fterminaln$ be a tie-breaking terminal
evaluation function. Let $\left(\Spartiel,\actions\right)$ be a game
tree built by $\umaxn$ or $\descenten$ from a certain state using
$\fterminal$. Let $s\in\Spartiel$. 

If $r\left(s\right)=1$ then there exists a unique value $\maxn$
of $s$with respect to $\fterminal$, denoted by $M\left(s\right)$,
and we have $\left(c\left(s\right),v\left(s\right)\right)=M\left(s\right)$.
\end{prop}

\begin{proof}
Let $\left(\Spartiel,\actions\right)$ be a game tree built by $\umaxn$
(resp. $\descenten$) from a certain state (i.e. the algorithm has
been applied $k$ times on that state). We show this property by induction.
Let $s\in\Spartiel$ such that $r\left(s\right)=1$. We first show
that this property holds for terminal states.

Suppose in addition that $\terminal\left(s\right)$ is true. Thus,
$\left(c\left(s\right),v\left(s\right)\right)=\left(\fbin\left(s\right),\fterminaln\left(s\right)\right)$
and therefore $\left(c\left(s\right),v\left(s\right)\right)=\left(\fbin\left(s\right),\fterminaln\left(s\right)\right)=M\left(s\right)$.

We now show this property for non-terminal states: we suppose instead
that $\terminal\left(s\right)$ is false.

Since $r\left(s\right)=1$, we have either for all $s'\in\actions\left(s\right)$,
$r\left(s'\right)=1$ or $c\left(s\right)_{\joueur\left(s\right)}=1$,
$v\left(s\right)_{\joueur\left(s\right)}=\max_{s'\in\liste{s'\in S}{\terminal\left(s'\right)}}\fterminaln\left(s'\right)_{\joueur\left(s\right)}$,
and there exists $\tilde{s}\in\actions\left(s\right)$ such that $c\left(s\right)=c\left(\tilde{s}\right)$,
$v\left(s\right)=v\left(\tilde{s}\right)$ and $r\left(\tilde{s}\right)=1$,
by Lemma~\ref{lem:completion-resolution-n-1-1-var}. 

If for all $s'\in\actions\left(s\right)$, $r\left(s'\right)=1$,
then by induction, we have for all $s'\in\actions\left(s\right)$,
$\left(c\left(s'\right),v\left(s'\right)\right)=M\left(s'\right)$.
But $c\left(s\right)=c\left(\argmax_{s'\in\actions\left(s\right)}\left(c\left(s'\right)_{\joueur\left(s\right)},v\left(s'\right)_{\joueur\left(s\right)}\right)\right)$
and $v\left(s\right)=v\left(\argmax_{s'\in\actions\left(s\right)}\left(c\left(s'\right)_{\joueur\left(s\right)},v\left(s'\right)_{\joueur\left(s\right)}\right)\right)$,
since there is a unique pair $\left(c\left(s'\right),v\left(s'\right)\right)$
maximizing $\left(c\left(s'\right)_{\joueur\left(s\right)},v\left(s'\right)_{\joueur\left(s\right)}\right)$
(as $\fterminaln$ is tie-breaking and that the values of a state
no longer change as soon as it is marked as solved). Therefore $\left(c\left(s\right),v\left(s\right)\right)=M\left(\argmax_{s'\in\actions\left(s\right)}\left(M\left(s'\right)_{0,\joueur\left(s\right)},M\left(s'\right)_{1,\joueur\left(s\right)}\right)\right)$,
hence $\left(c\left(s\right),v\left(s\right)\right)=M\left(s\right)$.

Suppose $c\left(s\right)_{\joueur\left(s\right)}=1$, $v\left(s\right)_{\joueur\left(s\right)}=\max_{s'\in\liste{s'\in S}{\terminal\left(s'\right)}}\fterminaln\left(s'\right)_{\joueur\left(s\right)}$,
and there exists $\tilde{s}\in\actions\left(s\right)$ such that $c\left(s\right)=c\left(\tilde{s}\right)$,
$v\left(s\right)=v\left(\tilde{s}\right)$ and $r\left(\tilde{s}\right)=1$.
By induction, $\left(c\left(\tilde{s}\right),v\left(\tilde{s}\right)\right)=M\left(\tilde{s}\right)$
and therefore $\left(c\left(s\right),v\left(s\right)\right)=M\left(\tilde{s}\right)$.
But, since $M\left(\tilde{s}\right)$ is maximum for the player $\joueur\left(s\right)$,
for all $s'\in\actions\left(s\right)$ either $\left(M\left(\tilde{s}\right)_{0,\joueur\left(s\right)},M\left(\tilde{s}\right)_{1,\joueur\left(s\right)}\right)>\left(M\left(s'\right)_{0,\joueur\left(s\right)},M\left(s'\right)_{1,\joueur\left(s\right)}\right)$
or $M\left(\tilde{s}\right)=M\left(s'\right)$ (as $\fterminaln$
is tie-breaking). Thus, $\left(c\left(s\right),v\left(s\right)\right)=M\left(\argmax_{s'\in\actions\left(s\right)}\left(M\left(s'\right)_{0,\joueur\left(s\right)},M\left(s'\right)_{1,\joueur\left(s\right)}\right)\right)$,
hence $\left(c\left(s\right),v\left(s\right)\right)=M\left(s\right)$.
\end{proof}
\begin{prop}
\label{prop:arret_n-var}Let $S$ be the set of states of a perfect
multi-player game. There exists $N\in\mathbb{N}$ such that after
applying $N$ times the algorithm $\descenten$ (resp. $\umaxn$)
on any state $s\in S$, we have $r\left(s\right)=1$.
\end{prop}

\begin{proof}
The proof is analogous to that of Proposition~\ref{prop:arret_n}.
\end{proof}
\begin{thm}
The algorithm $\descenten$ and the algorithm $\umaxn$ are ``complete'',
i.e. applying $\descenten$ (resp. $\umaxn$) on any state $s\in\etats$
by using a tie-breaking terminal evaluation $\fterminal$, with a
search time $\tau$ large enough, gives $r\left(s\right)=1$ and $\left(c\left(s\right),v\left(s\right)\right)=M\left(s\right)$,
the unique value $\maxn$ of $s$ with respect to $\fterminal$.
\end{thm}

\begin{proof}
By Proposition~\ref{prop:arret_n-var}, then by Proposition~\ref{prop:arret_n-var}.
\end{proof}

\bibliographystyle{plain}
\bibliography{proof_bf_dd_dag}

\begin{thebibliography}{10}

\bibitem{baier2020guiding}
Hendrik Baier and Michael Kaisers.
\newblock Guiding multiplayer mcts by focusing on yourself.
\newblock In {\em 2020 IEEE Conference on Games (CoG)}, pages 550--557. IEEE,
  2020.

\bibitem{baier2020opponent}
Hendrik Baier and Michael Kaisers.
\newblock Opponent-pruning paranoid search.
\newblock In {\em International Conference on the Foundations of Digital
  Games}, pages 1--7, 2020.

\bibitem{browne2012survey}
Cameron~B Browne, Edward Powley, Daniel Whitehouse, Simon~M Lucas, Peter~I
  Cowling, Philipp Rohlfshagen, Stephen Tavener, Diego Perez, Spyridon
  Samothrakis, and Simon Colton.
\newblock A survey of monte carlo tree search methods.
\newblock {\em Transactions on Computational Intelligence and AI in games},
  4(1):1--43, 2012.

\bibitem{cohen2019apprendre}
Quentin Cohen-Solal.
\newblock Apprendre {\`a} jouer aux jeux {\`a} deux joueurs {\`a} information
  parfaite sans connaissance.
\newblock 2019.

\bibitem{2020learning}
Quentin Cohen-Solal.
\newblock Learning to play two-player perfect-information games without
  knowledge.
\newblock {\em arXiv preprint arXiv:2008.01188}, 2020.

\bibitem{cohen2020minimax}
Quentin Cohen-Solal and Tristan Cazenave.
\newblock Minimax strikes back.
\newblock {\em arXiv preprint arXiv:2012.10700}, 2020.

\bibitem{Coulom06}
R{\'{e}}mi Coulom.
\newblock Efficient selectivity and backup operators in monte-carlo tree
  search.
\newblock In {\em Computers and Games, 5th International Conference, {CG} 2006,
  Turin, Italy, May 29-31, 2006. Revised Papers}, pages 72--83, 2006.

\bibitem{fridenfalk2014n}
Mikael Fridenfalk.
\newblock N-person minimax and alpha-beta pruning.
\newblock In {\em NICOGRAPH International 2014, Visby, Sweden, May 2014}, pages
  43--52, 2014.

\bibitem{kocsis2006bandit}
Levente Kocsis and Csaba Szepesv{\'a}ri.
\newblock Bandit based monte-carlo planning.
\newblock In {\em European conference on machine learning}, pages 282--293.
  Springer, 2006.

\bibitem{kocsis2006improved}
Levente Kocsis, Csaba Szepesv{\'a}ri, and Jan Willemson.
\newblock Improved monte-carlo search.
\newblock {\em Univ. Tartu, Estonia, Tech. Rep}, 1, 2006.

\bibitem{korf1991multi}
Richard~E Korf.
\newblock Multi-player alpha-beta pruning.
\newblock {\em Artificial Intelligence}, 48(1):99--111, 1991.

\bibitem{korf1996best}
Richard~E Korf and David~Maxwell Chickering.
\newblock Best-first minimax search.
\newblock {\em Artificial intelligence}, 84(1-2):299--337, 1996.

\bibitem{luckhart1986algorithmic}
Carol Luckhart and Keki~B Irani.
\newblock An algorithmic solution of n-person games.
\newblock In {\em AAAI}, volume~86, pages 158--162, 1986.

\bibitem{nijssen2012overview}
JAM Nijssen and Mark~HM Winands.
\newblock An overview of search techniques in multi-player games.
\newblock In {\em Computer Games Workshop at ECAI}, pages 50--61, 2012.

\bibitem{petosa2019multiplayer}
Nick Petosa and Tucker Balch.
\newblock Multiplayer alphazero.
\newblock {\em arXiv preprint arXiv:1910.13012}, 2019.

\bibitem{schadd2011best}
Maarten~PD Schadd and Mark~HM Winands.
\newblock Best reply search for multiplayer games.
\newblock {\em IEEE Transactions on Computational Intelligence and AI in
  Games}, 3(1):57--66, 2011.

\bibitem{silver2018general}
David Silver, Thomas Hubert, Julian Schrittwieser, Ioannis Antonoglou, Matthew
  Lai, Arthur Guez, Marc Lanctot, Laurent Sifre, Dharshan Kumaran, Thore
  Graepel, et~al.
\newblock A general reinforcement learning algorithm that masters chess, shogi,
  and go through self-play.
\newblock {\em Science}, 362(6419):1140--1144, 2018.

\bibitem{sturtevant2003last}
Nathan~R Sturtevant.
\newblock Last-branch and speculative pruning algorithms for max\^{} n.
\newblock In {\em IJCAI}, volume~3, pages 669--678. Citeseer, 2003.

\bibitem{sturtevant2000pruning}
Nathan~R Sturtevant and Richard~E Korf.
\newblock On pruning techniques for multi-player games.
\newblock {\em AAAI/IAAI}, 49:201--207, 2000.

\bibitem{zerbel2019multiagent}
Nicholas Zerbel and Logan Yliniemi.
\newblock Multiagent monte carlo tree search.
\newblock In {\em Proceedings of the 18th International Conference on
  Autonomous Agents and MultiAgent Systems}, pages 2309--2311, 2019.

\end{thebibliography}

\end{document}